\documentclass[onecolumn]{article}
 \usepackage{graphicx}          %
 \usepackage{amssymb}
 \usepackage{amsmath}
\usepackage{amsthm}

\usepackage{array}

 \usepackage{calc}
 \usepackage{tikz}
 \usepackage{pgfplots}
\usepackage{cases}

\usetikzlibrary{arrows}
\usetikzlibrary{shapes}
 \usetikzlibrary{calc}
\usetikzlibrary{decorations.pathreplacing}
\usetikzlibrary{arrows,positioning} 
\usetikzlibrary{pgfplots.groupplots}
\def\imagetop#1{\vtop{\null\hbox{#1}}}

\newtheorem{thm}{Theorem}

\newtheorem{lemma}{Lemma}

\newtheorem{assum}{Assumption}

\newtheorem{prop}{Proposition}

\theoremstyle{definition}
\usepackage{thmtools}
\declaretheorem[style=definition,qed=$\blacksquare$]{definition}

\newtheorem{example}{Example}

\newtheorem{corollary}{Corollary}

\newtheorem*{PPFIFO}{PP/FIFO Rule}

\newcommand{\densc}{\rho^\text{crit}}
\newcommand{\dens}{\rho}
\newcommand{\hdens}{\hat{\dens}}
\newcommand{\flux}{\Phi}
\newcommand{\fluxin}{\flux^\text{in}}
\newcommand{\fluxout}{\flux^\text{out}}

\newcommand{\fluxoutp}{\frac{d}{d \dens_l}\flux^{\text{out}}}
\newcommand{\fluxoutmax}{\overline{\flux}^\text{out}}
\newcommand{\fluxoutm}{\flux^\text{out,$m$}}
\newcommand{\fluxoutmmax}{\bar{\flux}^\text{out,$m$}}
\newcommand{\densjam}{\rho^\text{jam}}
\newcommand{\fluxc}{\flux^\text{crit}}
\newcommand{\network}{\mathcal{G}}
\newcommand{\Verts}{\mathcal{V}}

\newcommand{\Links}{\mathcal{L}}
\newcommand{\Lstart}{\mathcal{R}^\text{start}}
\newcommand{\OLinks}{\mathcal{O}}
\newcommand{\Lin}{\mathcal{L}^\text{in}}
\newcommand{\Lout}{\mathcal{L}^\text{out}}
\newcommand{\inflow}{d}
\newcommand{\flowe}{f^\text{e}}
\newcommand{\floweO}{\flowe_\OLinks}
\newcommand{\flowealt}{\tilde{f}^\text{e}}

\newcommand{\flowin}{f^\text{in}}
\newcommand{\flowout}{f^\text{out}}
\newcommand{\rflowout}{f^\text{out}}
\newcommand{\dense}{\rho^\text{e}}
\newcommand{\densealt}{\tilde{\rho}^\text{e}}

\newcommand{\Ramps}{\mathcal{R}}

\newcommand{\turn}{\beta}
\newcommand{\off}{\gamma}
\newcommand{\head}{{\sigma}}
\newcommand{\tail}{{\tau}}
\newcommand{\Routea}{A}
\newcommand{\Routeb}{B}
\newcommand{\Sink}{\Verts^\text{sink}}
\newcommand{\meter}{m}
\newcommand{\meteralt}{\tilde{m}}
\newcommand{\Domain}{\mathcal{D}}
\newcommand{\hDomain}{\hat{\Domain}}
\newcommand{\hOmega}{\hat{\Omega}}
\def\imagetop#1{\vtop{\null\hbox{#1}}}

\newcommand{\Ball}{\mathcal{B}}
\newcommand{\bJ}{\bar{J}}
\newcommand{\bF}{\bar{F}}
\newcommand{\bQ}{\bar{Q}}
\newcommand{\bV}{\bar{V}}

\title{A Compartmental Model for Traffic Networks and its Dynamical Behavior}
\date{}
\author{Samuel Coogan and Murat Arcak
\thanks{This work was supported in part by the National Science Foundation under grant ECCS-1101876 and by the Air Force Office of Scientific Research under grant FA9550-11-1-0244.}
\thanks{S. Coogan and M. Arcak are with the Department of Electrical Engineering and Computer Sciences, University of California, Berkeley, Berkeley, CA. \texttt{\{scoogan,arcak\}@eecs.berkeley.edu}}}
\renewcommand\footnotemark{}

\begin{document}
\maketitle

\tikzstyle{link}=[line width=2pt, ->,>=latex]
\tikzstyle{junc}=[draw,circle,inner sep=1pt,minimum width=8pt]
\tikzstyle{onramp}=[line width=2pt, dashed,->,>=latex]
\begin{abstract}
We propose a macroscopic traffic network flow model suitable for analysis as a dynamical system, and we qualitatively analyze equilibrium flows as well as convergence. Flows at a junction are determined by downstream \emph{supply} of capacity as well as upstream \emph{demand} of traffic wishing to flow through the junction.  This approach is rooted in the celebrated Cell Transmission Model for freeway traffic flow.  Unlike related results which rely on certain system cooperativity properties, our model generally does not possess these properties.  We show that the lack of cooperativity is in fact a useful feature that allows traffic control methods, such as ramp metering, to be effective. Finally, we leverage the results of the paper to develop a linear program for optimal ramp metering.

  \end{abstract}

\section{Introduction}

Despite the economic importance of mitigating traffic congestion \cite{kurzhanskiy2010active} and the large number of modeling approaches considered in the literature (see \cite{Helbing:2001uq, Hoogendoorn:2001uq} for reviews), few publications investigate the qualitative dynamical properties of traffic flow models for general network topologies.  For example, models such as \cite{Papageorgiou:1990ij, Messner:1990bs, Lebacque:1996zr} and the celebrated \emph{Cell Transmission Model (CTM)} of Daganzo \cite{Daganzo:1994fk, Daganzo:1995kx} were primarily developed for simulation with few analytical results available. The primary exception is \cite{Gomes:2008fk} which provides a thorough investigation of the CTM when modeling a stretch of highway.  The authors characterize equilibria and stability properties for this specific network class, but the results are not extended to more general networks, the authors assume a specific class of linear supply and demand functions, and the dynamics resulting from infeasible onramp demands are not fully analyzed.

We propose a general model that encompasses the CTM as defined in \cite{Daganzo:1994fk, Daganzo:1995kx, Gomes:2008fk} and extends the model to general nonlinear supply and demand functions and to more general network topologies. Using this model, we significantly extend the few existing results on equilibria and convergence such as \cite{Gomes:2008fk} and we present a simple linear program for obtaining a \emph{ramp metering} control strategy that achieves the maximum possible steady-state network throughput.

Our work is related to the \emph{dynamical flow networks} recently proposed in \cite{Como:2013ve, Como:2013bh} and further studied in \cite{Como:2013qf}.  In \cite{Como:2013ve, Como:2013bh}, downstream supply is not considered and thus downstream congestion does not affect upstream flow, an unrealistic assumption for traffic modeling. In  \cite{Como:2013qf}, the authors allow flow to depend on the density of downstream links, but the paper focuses on throughput optimality of a particular class of routing policies that ensure the resulting dynamics are \emph{cooperative} \cite{Hirsch:1985fk, Angeli:2003fv}. In contrast, the model proposed here is generally not cooperative. Furthermore, the adaptation to the CTM described briefly in \cite[Section II.C]{Como:2013qf} differs from our model in the following important respects: the model as discussed in \cite[Section II.C]{Como:2013qf} assumes a path graph network topology, requires identical links (\emph{i.e.}, identical supply and demand functions), and only considers trajectories in the region in which supply does not restrict flow (that is, $\alpha^v(\dens)=1$ for all $v\in\Verts$ in our model), which is shown to be positively invariant given their assumptions. In this work, we generalize each of these restrictions.   

In a separate direction of research, many network models attempt to apply single road PDE models such as \cite{Lighthill:1955vn, Richards:1956ys} directly to networks, see \cite{Garavello:2006tg} for a thorough treatment. Recent results such as \cite{Canepa:2012fk} and \cite{Han:2012ve} provide analytical tools for traffic network estimation and modeling using PDE models. The CTM and related models, including our proposed model, can be considered to be a discretization of an appropriate PDE model \cite{Lebacque:1996zr}. Alternatively, these models and the model we propose in this work fit into the broad class of \emph{compartmental systems} that model the flow of a substance among interconnected ``compartments'' \cite{Sandberg:1978fk, Maeda:1978fk,Jacquez:1993uq}.

We first proposed a compartmental model of traffic flow in \cite{Coogan:2014ph}. Here, we expand on the conference version by discussing the general lack of cooperativity for our proposed model, identifying how lack of cooperativity can be exploited to increase throughput \emph{via} ramp metering, and providing an explicit optimization problem for obtaining a ramp metering strategy that achieves the maximum possible network throughput.

In Section \ref{sec:model}, we propose the traffic network model.  In Section \ref{sec:mono}, we discuss conditions under which our model is and is not cooperative. In Section \ref{sec:equil}, we characterize existence and uniqueness of equilibrium flows. We demonstrate how the preceding analysis can be used for ramp metering in Section \ref{sec:ramp-metering}. 

\section{Dynamic Model of Traffic}
\label{sec:model}
\subsection{Network structure}

A traffic network consists of a directed graph $\network=(\Verts,\OLinks)$ with \emph{junctions} $\Verts$ and \emph{ordinary links} $\OLinks$ along with a set of \emph{onramps} $\Ramps$ which serve as entry points into the network.  For $l\in\OLinks$, let $\head(l)$ denote the head vertex of link $l$ and let $\tail(l)$ denote the tail vertex of link $l$, and traffic flows from $\tail(l)$ to $\head(l)$. Each onramp $l\in\Ramps$ directs an exogenous input flow onto $\network$ via a junction, and $\head(l)\in \Verts$ for $l\in\Ramps$ denotes the entry junction for onramp $l$. By convention, $\tail(l)=\emptyset$ for all $l\in\Ramps$. Ordinary links (resp., onramps) are denoted with a solid (resp., dashed) arrow in figures. 
\begin{assum}
  The traffic network graph is acyclic.
\end{assum}
Acyclicity is a reasonable assumption when modeling a portion of the road network of particular interest. For example, the road network leading out of a metropolitan area during the evening commute may be modeled as an acyclic graph where road links leading towards the metropolitan area are not modeled due to low utilization by commuters.

Let $\Links\triangleq\OLinks\cup \Ramps$.  For each $v\in\Verts$, we denote by $\Lin_v\subset \Links$ the set of incoming links to node $v$ and by $\Lout_v\subset \Links$ the set of outgoing links, \emph{i.e.},  $\Lin_v=\{l:\head(l)=v\}$ and $\Lout_v=\{l:\tail(l)=v\}$.
We assume $\Lin_v\neq \emptyset$ for all $v\in \Verts$, thus the network flows start at onramps. Furthermore, we assume $\Lout_{\head(l)}\neq \emptyset$ for all $l\in\Ramps$, \emph{i.e.}, onramps always flow into at least one ordinary link downstream.

We define $  \Lstart\triangleq \{l\in\Ramps:\Lin_{\head(l)}\cap \OLinks=\emptyset\}$ to be the set of links that lead to junctions that have only onramps as incoming links, and $\Sink \triangleq \{v\in \Verts: \Lout_v=\emptyset\}$ to be the set of junctions that have no outgoing links. %

\subsection{Link supply and demand}
For each link $l\in \OLinks$, we associate the time-varying density $\dens_l(t)\in[0,\densjam_l]$ where $\densjam_l\in(0,\infty)$ is the \emph{jam density} of link $l$. For $l\in \Ramps$, we associate the time-varying density $\dens_l(t)\in[0,\infty)$, thus onramps have no maximum density, that is, they act as ``queues''. We define $\dens\triangleq\{\dens_l\}_{l\in\Links}$. %

Furthermore, we assume each $l\in\Links$ possesses a \emph{demand} function $\fluxout_l(\dens_l)$  that quantifies the amount of traffic wishing to flow downstream, and we assume each $l\in\OLinks$ possesses a \emph{supply} function $\fluxin_l(\dens_l)$. We make the following assumption on the supply and demand functions:
\begin{assum}
\label{assum:2}
  For each $l\in\OLinks$:%
  \begin{enumerate}
\renewcommand{\labelenumi}{A\arabic{enumi}. }
  \item The demand function $\fluxout_l(\dens_l):[0,\densjam_l]\to \mathbb{R}_{\geq 0}$ is strictly increasing and continuously differentiable\footnote{These assumptions are made to simplify the exposition but can be relaxed to Lipschitz continuity and nonstrict monotonicty beyond the critical density; such functions are considered in the examples.} on $(0,\densjam_l)$ with $\fluxout_l(0)=0$, and $\frac{d }{d\dens_l}\fluxout_l(\dens_l)$ is bounded above.
  \item The supply function $\fluxin_l(\dens_l):[0,\densjam_l]\to \mathbb{R}_{\geq 0}$ is strictly decreasing and continuously differentiable on $(0,\densjam_l)$ with $\fluxin_l(\densjam_l)=0$, and $\frac{d }{d\dens_l}\fluxin_l(\dens_l)$ is bounded below.
  \end{enumerate}
For each $\l\in\Ramps$:
\begin{enumerate}
 \setcounter{enumi}{2}
\renewcommand{\labelenumi}{A\arabic{enumi}. }
\item In addition to A1, $\fluxout_l(\dens_l)$ is bounded above with supremum $\fluxoutmax_l\triangleq\sup\fluxout_l(\dens_l)$ and there exists $M_l>0$ such that\footnote{The bound on the derivative of $\fluxout_l(\dens_l)$ is a very mild technical condition used in the proofs of some propositions. For example, the condition is satisfied when $\fluxout_l(\dens_l)$ attains its maximum.}  $\frac{d}{d\dens_l}\fluxout_l(\dens_l)\leq M_l(1+\dens_l)^{-2}$ for all $\dens_l$.
\end{enumerate}
\end{assum}

Assumption \ref{assum:2} implies that for each $l\in\OLinks$, there exists unique $\densc_l$ such that $\fluxout_l(\densc_l)=\fluxin_l(\densc_l)=:\fluxc_l$. 
Fig. \ref{fig:1} depicts examples of supply and demand functions.%
\begin{figure}
  \centering
\begin{tabular}{c c}
\imagetop{\begin{tikzpicture}
  \begin{axis}[width=.4\textwidth,axis x line=center, axis y line=center, xmin=0, xmax=5, ymax=3.5,ymin=0,tick align=outside,     xtick={1.0362,3.5}, xticklabel shift=-2pt, unit vector ratio*=.8 .8 .8,
    xticklabels={$\densc_l$,$\densjam$}, ytick={1.1270},yticklabels={$\fluxc_l$},
    xlabel=$\dens_l$,
    ylabel=$\flux_l$,
legend style={at={(axis cs: 2.5,-1.6)}, anchor=north}
]
    \addplot+[mark=none,smooth, black] (\x,{2-2*exp(-.8*\x)});
\addlegendentry{$\fluxout_l(\dens_l)$}
\addplot+[domain=0:3.5,mark=none,smooth,black,dashed](\x,{3*exp(-.8*\x)-3*exp(-.8*3.5)});
\addlegendentry{$\fluxin_l(\dens_l)$}
  \end{axis}
\end{tikzpicture}}&
\imagetop{\begin{tikzpicture}
  \begin{axis}[width=.4\textwidth,axis x line=center, axis y line=center, xmin=0, xmax=5, ymax=3.5,ymin=0,tick align=outside, xtick=\empty,unit vector ratio*=.8 .8 .8, 
ytick={2},yticklabels={$\fluxoutmax_l$},
    xlabel=$\dens_l$,
    ylabel=$\flux_l$,
legend style={at={(axis cs: 2.5,-1.6)}, anchor=north}
]
    \addplot+[mark=none,smooth, black] (\x,{2-2*exp(-.8*\x)});
\addlegendentry{$\fluxout_l(\dens_l)$}
\addplot+[mark=none,smooth,black,dotted](\x,2);
  \end{axis}
\end{tikzpicture}}\\
(a) Ordinary link, $l\in\OLinks$&(b) Onramp link, $l\in\Ramps$
\end{tabular}
\caption{Plot of prototypical supply and demand functions $\fluxin(\dens)$ and $\fluxout(\dens)$ for (a) an ordinary road link, and prototypical demand function $\fluxout(\dens)$ for (b) an onramp link. }

  \label{fig:1}
\end{figure}

\subsection{Dynamic Model}
\label{sec:dyn_mod}
We now describe the time evolution of the densities on each link. The domain of interest is
\begin{align}
  \label{eq:61}
  \Domain\triangleq\{\dens:\dens_l\in[0,\infty) \ \forall l\in\Ramps \text{ and }\dens_l\in[0,\densjam_l] \ \forall l\in\OLinks\}.
\end{align}
Let $\Domain^\circ$ denote the interior of $\Domain$.

 For each onramp $l\in\Ramps$, we assume there exists exogenous \emph{input flow} $\inflow_l(t)$. Furthermore, for each $l\in\Links$ we subsequently define an output flow function $\rflowout_l(\dens)$, and for each $l\in\OLinks$ we define an input flow function $\flowin_l(\dens)$, such that 
\begin{alignat}{2}
\label{eq:6all}
\dot{\dens}_l=   F_l(\dens, t)\triangleq \begin{cases}\inflow_l(t)-\rflowout_l(\dens)   &\text{if } l\in\Ramps\\
 \flowin_l(\dens)-\flowout_l(\dens)  &\text{if } l\in\OLinks
\end{cases}
\end{alignat}
where the functions $\flowin_l(\dens)$ and $\flowout_l(\dens)$ are defined below.
When $\inflow_l(t)\equiv \inflow_l$ for constant $\inflow_l$ for all $l\in\Ramps$, the dynamics are autonomous and we write $F_l(\dens)$ instead. We define $F(\dens,t)\triangleq \begin{bmatrix}F_1(\dens,t)&\cdots F_{|\Links|}(\dens,t)\end{bmatrix}'$ for some enumeration of $|\Links|$ where $'$ denotes transpose, and we similarly define $F(\dens)$ when the dynamics are autonomous.

For each $l,k\in\Links$,
\begin{align}
  \label{eq:7}
  \turn_{lk}\in[0,1]%
\end{align}
is the \emph{split} ratio describing the fraction of vehicles flowing out of link $l$ that are routed to link $k$. It follows that $\beta_{lk}>0$ only if $\head(l)=\tail(k)$. We require that  $\sum_{k\in\Lout_v}\turn_{lk}\leq 1$ for all $l\in\Lin_v$
and   $ 1-\sum_{k\in\Lout_{\head(l)}}\turn_{lk}$ is interpreted to be the fraction of the outflow on link $l$ that is routed \emph{off} the network via, \emph{e.g.}, an infinite capacity offramp. To ensure continuity of $\flowout_l(\cdot)$, we make the following assumption:
\begin{assum}
\label{assum:posturn} If $v\not \in \Sink$, then $\turn_{lk}>0$ for all $l\in\Lin_v$ and all $k\in\Lout_v$.
\end{assum}

A large variety of phenomenological rules for determining the outflows of road links have been proposed in the literature; see \cite{Lebacque:2005fk, Lebacque:2004vn} for several examples.  We employ the \emph{proportional priority, first-in-first-out (PP/FIFO)} rule for junctions adapted from \cite{kurzhanskiy2010active}:
\begin{PPFIFO}
For $v\in\Sink$, $\flowout_l(\dens)\triangleq\fluxout_l(\dens_l)$ for all $l\in\Lin_v.$ For each $v\in\Verts\backslash \Sink$, we must ensure that the inflow of each outgoing link does not exceed the link supply. Define
\begin{alignat}{2}
  \label{eq:11-2}
  {\alpha}^v(\dens)\triangleq&&\max_{\alpha\in[0,1]}\quad&\alpha\\
\label{eq:11-3}&&\text{s.t.}\quad&\alpha\sum_{j\in\Lin_v}\turn_{jk} \fluxout_{j}(\dens_{j})\leq \fluxin_k(\dens_k)\quad \forall k\in\Lout_v. 
\end{alignat}
By scaling the demand of each link by $\alpha^v(\dens)$, we ensure that the supply of each downstream link is not violated:
  \begin{alignat}{2}
    \label{eq:11}
    \flowout_l(\dens)\triangleq\alpha^v(\dens)\fluxout_l(\dens_l)\quad \forall{l\in\Lin_v}.
  \end{alignat}
To complete the model, we determine $\flowin_l(\dens)$ from conservation of flow: 
\begin{align}
  \label{eq:9}
  \flowin_{l}(\dens)=\sum_{k\in \Lin_{\tail(l)}}\turn_{kl}\flowout_{k}(\dens)\quad \forall l\in\OLinks.
\end{align}
\end{PPFIFO}

The format of \eqref{eq:11-2} emphasizes the fact that the outflow of a link is the largest possible flow such that neither link demand nor downstream supply is exceeded and such that the outflow of all incoming links at a junction is proportional to the demand of these links. This proportionality constraint gives rise to the \emph{proportional priority} terminology. The fixed turn ratios along with the supply and demand restrictions implies that a lack of supply of an outgoing link restricts flow to other outgoing links, a phenomenon known in the transportation literature as a \emph{first-in-first-out (FIFO)} property \cite{kurzhanskiy2010active, Daganzo:1995kx}.

\subsection{Basic Properties of the PP/FIFO rule}
We first note two properties captured by the proposed network flow model.
\begin{lemma}
\label{lem:simp}
  A simple consequence of the PP/FIFO rule is for all $\dens\in\Domain$,
  \begin{alignat}{2}
    \label{eq:24}
    \flowin_l(\dens)&\leq \fluxin_l(\dens_l)&\qquad &\forall l\in\OLinks\\
    \label{eq:24-2}    \flowout_l(\dens)&\leq \fluxout_l(\dens_l)&&\forall l\in\Links.
  \end{alignat}
\end{lemma}
We note that the domain $\Domain$ in \eqref{eq:61} is easily seen to be positively invariant. %
Furthermore, it is not difficult to establish Lipschitz continuity of \eqref{eq:6all} which ensures global existence and uniqueness of solutions for piecewise continuous input flows $\{d_l(t)\}_{l\in\Ramps}$ \cite[Chapter I]{Hale:1980fk}.
Now suppose $d_l(t)\equiv d_l$ for some constant $d_l$ for all $l\in\Ramps$ so that the dynamics are autonomous. From the PP/FIFO rule, we conclude that $\flowout_l(\dens)$, and thus $F_l(\dens)$, is a continuous selection of differentiable functions determined by the constraints \eqref{eq:11-3}, that is, $F(\dens)$ is \emph{piecewise differentiable} \cite [Section 4.1]{Scholtes:2012fk}.

\section{Lack of cooperativity and its advantages}
\label{sec:mono}
The traffic network with constant input flows is \emph{cooperative} \cite{Hirsch:1985fk} if, for all links $l\in\Links$ and $\dens\in\Domain^{\circ}$, we have
\begin{align}
  \label{eq:62}
  \frac{\partial F_l(\dens)}{\partial \dens_k}\geq 0 \quad \forall k\neq l%
\end{align}
where, when $F_l(\dens)$ is not differentiable (which occurs on a set of measure zero), we interpret the partial derivative in an appropriate directional sense, see \cite[Scetion 4.1.2]{Scholtes:2012fk} for details.
Cooperative systems are order preserving systems with respect to the standard order defined by the positive orthant and are a special class of \emph{monotone} systems \cite{Hirsch:1985fk,Angeli:2003fv}.
We show that traffic networks are, in general, not cooperative. %
In the following example, increased demand of an incoming link at a junction causes a decrease in the inflow entering an outgoing link.

\begin{example}
\label{ex:2}
Consider a road network with two onramps labeled $\{1,2\}$ and two ordinary links labeled $\{3,4\}$ as shown in Fig. \ref{fig:nonmonotone}(b). 
Suppose that $\turn_{13}=\turn_{14}=\frac{1}{2}$, $\turn_{23}=\frac{2}{3}$, and $\turn_{24}=\frac{1}{3}$
and $\{v_2, v_3\}= \Sink$. Furthermore, suppose $\fluxout_i(\rho_i)=\max\{\rho_i,c \} $, for $i=1,2$ and $c\in\mathbb{R}_{>0}$.  Now consider $\dens=(\dens_1,\dens_2,\dens_3,\dens_4)$ such that $\dens_1=\dens_2<\frac{2}{9}c$ and suppose 
the supply of link 3 is the limiting factor for the flow through junction $v_1$ so that $  \flowin_3(\dens)= \fluxin_3(\dens)=\alpha^{v_1}(\dens)\left(\frac{1}{2}\dens_1+\frac{2}{3}\dens_2\right)=\frac{7}{6}\alpha^{v_1}(\dens)\dens_1$ and $\flowin_4(\dens)=\frac{5}{6}\alpha^{v_1}(\dens)\dens_1$. Now consider $\bar{\dens}=(\bar{\dens}_1,\bar{\dens} _2,\bar{\dens} _3,\bar{\dens} _4)\triangleq(\dens_1,\frac{9}{2}\dens_2,\dens_3,\dens_4)$. The supply of link 3 is unchanged, but the total demand from links 1 and 2 for link 3 has tripled so that $\alpha^{v_1}(\bar{\dens})=\frac{1}{3}\alpha^{v_1}(\dens)$. Then
\begin{align}
  \label{eq:142}
\textstyle  \flowin_4(\bar{\dens})=\alpha^{v_1}(\bar{\dens})\left(\frac{1}{2}\bar{\dens}_1+\frac{1}{3}\bar{\dens}_2\right)=\frac{2}{3}\alpha^{v_1}(\dens)\dens_1<\flowin_4(\dens)
\end{align}

Since $\flowout_4(\dens)=\flowout_4(\bar{\dens})=\fluxout_4(\dens_4)$, we have   $\dot{\bar{\dens}}_4(0)<\dot{\dens}_4(0)$ and thus there exists $\epsilon>0$ such that $\dens_4(\epsilon)>\bar{\dens}_4(\epsilon)$, showing the system is not cooperative.

 \begin{figure}
    \centering
\begin{tabular}{>{\centering\arraybackslash}m{3in}  >{\centering\arraybackslash}m{2in}}
    \begin{tikzpicture}[scale=.8]
\node (j0) at (0,0) {};
\node[junc] (j1) at ($(j0)+(0:.8in)$) {$v_1$};
\node[junc] (j2) at ($(j1)+(0:.8in)$) {$v_2$};
\node[junc] (j3) at ($(j2)+(0:.8in)$) {$v_3$};
\node[junc] (j4) at ($(j3)+(0:.8in)$) {$v_4$};
\draw[onramp] (j0)--(j1);
\draw[link] (j1)--(j2);
\draw[link] (j2)--(j3);
\draw[link] (j3)--(j4);
\draw[onramp] ($(j1)+(210:.8in)$)--(j1);
\draw[onramp] ($(j2)+(210:.8in)$)--(j2);
\draw[onramp] ($(j3)+(210:.8in)$)--(j3);
    \end{tikzpicture}\\
(a)\\
    \begin{tikzpicture}[scale=.8]
\node[junc] (j1) at (0,0) {$v_1$};
\draw[onramp] ($(j1)+(150:.8in)$) -- node[above]{$1$} (j1);
\draw[onramp] ($(j1)+(210:.8in)$) -- node[above]{$2$} (j1);
\node[junc] (j2) at ($(j1)+(30:.8in)$) {$v_2$};
\node[junc] (j3) at ($(j1)+(-30:.8in)$) {$v_3$};
\draw[link] (j1) -- node[above]{3}(j2);
\draw[link] (j1) -- node[above]{4}(j3);
    \end{tikzpicture}\\
(b)
\end{tabular}
    \caption{ (a) A network that models a stretch of highway with onramps. Each junction is such that $|\Lout_v|\leq 1$, \emph{i.e.}, each junction is a merge. ``Offramps'' are only modeled through the split ratios at junctions. This system is cooperative. (b) A network with two onramps $\{1,2\}$ and two ordinary links $\{3,4\}$. This example system is not cooperative due to the proportional priority assumption at junction $v_1$. In particular, increased density on link $2$ can decrease the flow entering link $4$.}
    \label{fig:nonmonotone}
  \end{figure}
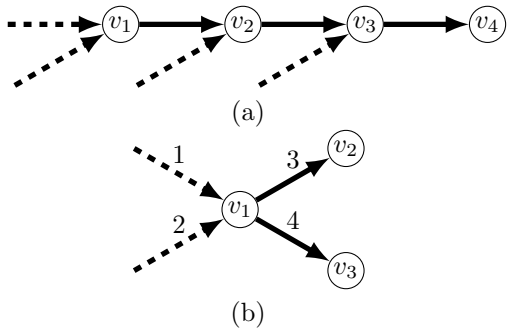
\end{example}
 Far from being a negative property of the model, lack of cooperativity is the main reason why ramp metering can increase network throughput or decrease average travel time. By metering the outflow of onramp $2$ in Example \ref{ex:2}, it would be possible to increase the inflow to link $4$, thereby increasing throughput.

\section{Equilibria and stability with constant input flows}
\label{sec:equil}
We now characterize the equilibria possible from the above model with constant input flow $\{\inflow_l\}_{l\in \Ramps}$. We will investigate the case where $\lim_{t\to\infty}F_l(\dens(t))= 0$ for all $l\in\Links$, and, when input flow exceeds network capacity, the case where $\lim_{t\to\infty}F_l(\dens(t))= 0$ for all $l\in\OLinks$ and $\lim_{t\to\infty}\flowout_l(\dens(t))=c_l\leq d_l$ for some constant $c_l$ for all $l\in\Ramps$. In the latter case, the density of some onramps (specifically, those with $c_l<d_l$) will diverge to infinity, but we will see that a meaningful definition of equilibrium nonetheless exists. From a practical point of view, such a characterization is useful, \emph{e.g.}, during ``rush hour'' when the input flow of a traffic network may exceed network capacity for a limited but extended period of time.

Define
$  \flowout_{\Ramps}(\dens)\triangleq
\label{eq:13-end}  \begin{bmatrix}
    \flowout_1(\dens)&\ldots&    \flowout_{|\Ramps|}(\dens)\end{bmatrix}'$, and likewise for $ \flowin_{\OLinks}(\dens)$ and $ \flowout_{\Ramps}(\dens)$
for some enumeration of $\OLinks$ and $\Ramps$. The dynamics \eqref{eq:9}  have the form
\begin{align}
  \label{eq:14}
  \flowin_{\OLinks}(\dens)=\Routea   \flowout_{\OLinks}(\dens)+\Routeb\flowout_{\Ramps}(\dens).
\end{align}
where $\Routea_{lk}=\beta_{kl}$ for $l,k\in\OLinks$ and $\Routeb_{lk}=\beta_{kl}$ for $l\in\OLinks$, $k\in\Ramps$. Acyclicity ensures $(I-A)$ is invertible: this can be seen by noting that the only solution to the equation $f=Af$ is $f=0$, which follows by a cascading argument since $f_l=0$ for $l\in L_1\triangleq \{l\in\OLinks \mid \Lin_{\tail(l)}\subset \Ramps\}$ since $A_{lk}=0$ for all $k$ for $l\in L_1$, then $f_l=0$ for $l\in L_2\triangleq \{l\in\OLinks\mid \Lin_{\tail(l)}\subset L_1\cup \Ramps\}$ since $A_{lk}=0$ for $k\not \in L_1$ for $l\in L_2$, \emph{etc.}  That is, a nonzero solution implies vehicles remain in the network indefinitely, implying existence of a cycle in the network.

\subsection{Feasible input flows}
\begin{definition}
\label{def:equil1}
The constant input flow $\{\inflow_l\}_{l\in\Ramps}$ is \emph{feasible} if there exists density $\dense\triangleq\{\dense_l\}_{l\in\Links}\in\Domain$ such that
\begin{alignat}{2}
  \label{eq:17}
\flowout_l(\dense)&=d_l&&\forall l\in\Ramps\\
\label{eq:17-2}\flowout_l(\dense)&=\flowin_l(\dense) \qquad &&\forall l\in\OLinks.
\end{alignat}
We define $\flowe_l\triangleq \flowout_l(\dense)$ for all $l\in \Links$, and the set $\{\flowe_l\}_{l\in\Links}$ is called an \emph{equilibrium flow}.
\end{definition}
If the input flow is not feasible, it is said to be \emph{infeasible}. It is clear that for a feasible input flow $\{\inflow_l\}_{l\in\Ramps}$, we must have for all $l\in\Ramps$:
  \begin{align}
    \label{eq:2}
\inflow_l&\leq \fluxoutmax_l \qquad \ \substack{\text{\small if there exists $\dens_l^*<\infty$ such that}\\\text{\small $\fluxout_l(\dens_l)=\fluxoutmax_l$ for all $\dens_l\geq \dens_l^*$}}\\
\label{eq:2-2}\text{or}\quad\inflow_l&< \fluxoutmax_l \qquad \text{ if $\fluxout_l(\dens_l)<\fluxoutmax_l$ for all $\dens_l\in[0,\infty)$}.
  \end{align}

\begin{prop}
\label{prop:lessthancrit}
An equilibrium flow $\{\flowe_l\}_{l\in\Links}$ with corresponding equilibrium densities $\{\dense_l\}_{l\in\Links}$ satisfies
  \begin{align}
    \label{eq:10}
\flowe_l\leq \fluxc_l\qquad \forall l\in\OLinks.
  \end{align}
\end{prop}
\begin{proof}
  Suppose there exists $l\in\OLinks$ such that $\flowe_l> \fluxc_l$. By the definition of equilibrium flow, we have $\flowout_l(\dense)=\flowin_l(\dense)=\flowe_l$. Since $\flowout_l(\dens)\leq\fluxout_l(\dens_l)$ for all $\dens$ by \eqref{eq:24}, we have $\fluxout_l(\dense_l)>\fluxc_l$. But by Assumption \ref{assum:2}, for all $\dens_l$ such that $\fluxout_l(\dens_l)>\fluxc_l$, it must be $\fluxin_l(\dens_l)\leq \fluxc_l$ and thus $\flowin_l(\dense)=\flowe_l>\fluxin_l(\dens_l)$, which contradicts \eqref{eq:24-2}.
\end{proof}

\begin{prop}
\label{prop:freeflow}
Assume \eqref{eq:2}--\eqref{eq:2-2}.  An input flow $\{\inflow_l\}_{l\in\Ramps}$ is feasible if and only if
  \begin{align}
  \label{eq:18}
    (I-\Routea)^{-1}\Routeb d\leq \fluxc
  \end{align}
where $d\triangleq\begin{bmatrix}d_1&\ldots&d_{|\Ramps|}\end{bmatrix}'$, $\fluxc \triangleq
\begin{bmatrix}
  \fluxc_1&\cdots&\fluxc_{|\mathcal{O}|}
\end{bmatrix}'$, and $\leq$ denotes elementwise inequality.
Furthermore, for feasible input flows, the equilibrium flow $\{\flowe_l\}_{l\in\OLinks}$ is unique.
\end{prop}
\begin{proof}
(uniqueness) By \eqref{eq:14}, an equilibrium flow of a feasible input flow satisfies
  \begin{align}
    \label{new_eq:19}
    \floweO=\Routea \floweO+\Routeb d
  \end{align}
where $\floweO$ is the vector of equilibrium flows for $\OLinks$.  Thus $\floweO=(I-\Routea)^{-1}\Routeb d$ is the unique solution to \eqref{new_eq:19}. 

  (only if) Applying Proposition \ref{prop:lessthancrit} to the unique $\floweO$ above gives necessity.

(if) Let $\floweO=(I-\Routea)^{-1}\Routeb d$ be a candidate equilibrium flow. From \eqref{eq:18}, there exists unique $\{\dense_l\}_{l\in\OLinks}$ such that $\flowe_l=\fluxout_l(\dense_l)$ for which $\flowe_l\leq \fluxin_l(\dense_l)$ for all $\l\in\OLinks$. Furthermore, there exists $\{\dense_l\}_{l\in\Ramps}$ such that $\fluxout_{l}(\dense_l)=d_l$ for all $\l\in\Ramps$ by \eqref{eq:2}--\eqref{eq:2-2}.  We now show that these flows satisfy the PP/FIFO rule. We first show that $\alpha^v(\dense)=1$. Considering \eqref{eq:11-2}--\eqref{eq:11-3}, for all $v\in\Links\backslash \Sink$:
\begin{alignat}{2}
  \label{eq:20}
\sum_{l\in\Lin_v}\turn_{lk}^v\fluxout_l(\dense_l)&=\sum_{l\in(\Lin_v\cap \OLinks)}\turn_{lk}^v\flowe_l+\sum_{l\in(\Lin_v\cap \Ramps)}\turn_{lk}^vd_l&\qquad &\forall k\in\Lout_v\\
\label{eq:20-2}&=\flowe_k&\qquad &\forall k\in\Lout_v\\
&\leq \fluxin_k(\dense_k)&\qquad &\forall k\in\Lout_v
\end{alignat}
where \eqref{eq:20-2} follows from \eqref{new_eq:19}, and thus  \eqref{eq:11-2}--\eqref{eq:11} is satisfied when $\alpha^v(\dense)=1$ for all $v\in\Verts$.
Also, $\flowout_l(\dens)=\fluxout_l(\dens_l)$ for all $l\in \Lin_v$ for all $v\in\Sink$ and \eqref{eq:9} follows from \eqref{new_eq:19}, thus proving sufficiency.
\end{proof}
While the equilibrium flow for a feasible input flow is unique by Proposition \ref{prop:freeflow}, in general, multiple equilibria densities may support this equilibrium flow. However there does exists a unique equilibrium for which each link is in \emph{freeflow}:
\begin{definition}
  An ordinary link $l\in\OLinks$ is said to be in \emph{freeflow} if $\flowout_l(\dens)=\fluxout_l(\dens_l)$. Otherwise, link $l$ is \emph{congested}.
\end{definition}
Note that at equilibrium, if link $l\in\OLinks$ is in freeflow, then necessarily $\dens_l\leq \densc_l$.
\begin{corollary}
\label{cor:freeflow_equil}
For a feasible input flow, there exists a unique equilibrium density $\{\dense_l\}_{l\in\Links}$ such that each link $\l\in\OLinks$ is in freeflow.
\end{corollary}

\begin{proof}
  Such an equilibrium density is constructed in the proof of the ``if'' direction in the proof of Proposition \ref{prop:freeflow}.
\end{proof}

Furthermore, if the input flow is \emph{strictly feasible}, then the equilibrium density is unique, and, moreover, it is asymptotically stable:
\begin{definition}
A feasible input flow $\{\inflow_l\}_{l\in\Ramps}$ is said to be \emph{strictly feasible} if the corresponding (unique) equilibrium flow satisfies $\flowe_l<\fluxc_l$ for all $l\in\OLinks$.  
\end{definition}
\begin{prop}
\label{prop:strictfeas1}
If the input flow $\{\inflow_l\}_{l\in\Ramps}$ is strictly feasible, then the equilibrium density is unique and it is locally asymptotically stable.
\end{prop}

\begin{proof}
(uniqueness) We claim if $\flowe_l<\fluxout_l(\dense_l)$ for any $l\in\Links$, then there exists $k\in \OLinks$ such that $\flowe_k=\fluxc_k$. For now we take this claim to be true and note that it contradicts the hypothesis of strictly feasible flows. Thus we conclude $\flowe_l=\fluxout_l(\dense_l)$ for all $l\in\Links$. Furthermore, since $\fluxout_l(\cdot)$ is strictly increasing, $\dense_l$ is unique.

To prove the claim, suppose $\flowe_l<\fluxout_l(\dense_l)$ for some $l\in\Links$ and $\flowe_k<\fluxc_k$ for all $k\in \OLinks$. There must exist $l'\in\Lout_{\head(l)}$ such that $\flowe_{l'}=\fluxin_{l'}(\dense_{l'})$, \emph{i.e.} there exists a lack of supply on link $l'$ since the flow on link $l$ is less than demand. Since $\flowe_{l'}<\fluxc_{l'}$ by assumption, we then have $\flowe_{l'}<\fluxout_{l'}(\dense_{l'})$ and we find another $l''\in\Lout_{\head(l')}$ such that $\flowe_{l''}=\fluxin_{l''}(\dense_{l''})$, but this cannot continue indefinitely since the traffic network is acyclic and finite, thus there exists $k\in \OLinks$ such that $\flowe_k=\fluxc_k$.

(stability) Because the traffic network is directed and acyclic, it is a standard graph theoretic result that there exists a topological ordering on the junctions. From this topological ordering on the junctions, we enumerate the links with the numbering function $e(\cdot):\Links\to\{0,\ldots,|\Links|\}$ where $e(l)$ is the enumeration of link $l$ such that for each $v$ and $l\in\Lin_v$, $k\in\Lout_v$, we have $e(l)<e(k)$ (Such an enumeration can be accomplished by, \emph{e.g.}, first enumerating $\Lin_{v_1}$, then $\Lin_{v_2}$, \emph{etc.} where $v_1,\ldots,v_{|\Verts|}$ are the junctions in topological order). Furthermore, $\flowout_l(\dense)=\fluxout_l(\dense_l)$ for all $l$, and it is straightforward to show
  \begin{align}
    \label{eq:57}
    \frac{\partial F_l}{\partial \dens_l}(\dense)<0&\qquad\forall{l\in\Links}
  \end{align}
and
\begin{align}
  \label{eq:58}
    \frac{\partial F_l}{\partial \dens_k}(\dense)=0&\qquad \text{if }k\not \in \Lin_{\tail(l)},
\end{align}
thus the Jacobian matrix of the traffic network flow evaluated at the equilibrium, $\frac{\partial F}{\partial \dens}(\dense)$, is lower triangular with strictly negative entries along the diagonal and therefore Hurwitz. Asymptotic stability within a neighborhood of the equilibrium follows from, \emph{e.g.}, \cite[Theorem 4.7]{khalil}.
\end{proof}

Note that Corollary \ref{cor:freeflow_equil} implies each link $l\in\OLinks$ is in freeflow for the unique equilibrium in Proposition \ref{prop:strictfeas1}. 
Stability in Proposition \ref{prop:strictfeas1} can also proved using the Lyapunov function $||F(\dens)||_1$. We remark that the dynamics are cooperative when all links are in free-flow. This fact allows us to conclude convergence to the equilibrium from the invariant box $\{\dens:0\leq\dens_l\leq\dense_l\ \forall l\in\Links\}$, and the argument can be extended to (not necessarily strictly) feasible
input flows as in the following:
  \begin{prop}
    \label{prop:roc}
    For a feasible input flow, all trajectories $\dens(t)$ such that $0\leq \dens_l(0)\leq \dense_l$ for all $l\in\Links$ converge to $\{\dense_l\}_{l\in\Links}$ where $\{\dense_l\}_{l\in\Links}$ is the unique equilibrium density in Corollary \ref{cor:freeflow_equil} for which all links $l\in\OLinks$ are in freeflow. That is,
    \begin{align}
      \label{eq:80}
      \lim_{t\to\infty} \dens_l(t)=\dense_l\qquad \text{ if \quad}0\leq \dens_l(0)\leq \dense_l\ \forall l\in\Links.
    \end{align}
  \end{prop}

\begin{proof}

Let $\mathcal{A}\triangleq\{\dens:0\leq \dens_l\leq \dense_l \ \forall l\in\Links\}$ and note that $\alpha^v(\dens)=1$ for all $\dens\in\mathcal{A}$, $v\in\Verts$ since $\mathcal{A}$ is contained within the freeflow region. In particular, $\flowout_l(\dens)=\fluxout_l(\dens_l)$ for all $l\in\Links$ for $\dens\in\mathcal{A}$.

We now consider the general scalar system $\dot{x}=s(t)-g(x)$, $x(t)\in\mathbb{R}$ with $s(\cdot),g(\cdot)$ differentiable and monotone increasing in $t$ and $x$, respectively. We claim that if $\dot{x}(0)\geq 0$, then $\dot{x}(t)\geq 0$ for all $t$. Indeed, suppose $\dot{x}(\tau)=0$ for some $\tau\geq 0$. Then $\left.\frac{d}{dt}\dot{x}\right|_{t=\tau}=\left.\big(\dot{s}(t)+g'(x)\dot{x}\big)\right|_{t=\tau}=\dot{s}(\tau)\geq 0$.

Now consider $l\in \Lstart$. It follows that if $\dens_l(0)=0$, then $\dot{\dens}_l= d_l-\fluxout_l(\dens_l)$ and $\dot{\dens}_l(0)\geq 0$. From the above analysis, we have $\dot{\dens}_l\geq 0$ for all $t\geq 0$. Furthermore, $\flowout_l(\dens(t))=\fluxout_l(\dens_l(t))$ is monotonically increasing as a function of $t$. Since $\dot{\dens}_l=\flowin_l(\dens)-\fluxout_l(\dens_l)$ where $\flowin_l(\dens)=\sum_{k\in\Lin_{\tail(l)}}\turn^{\tail(l)}_{kl}\fluxout_k(\dens_k)$, we proceed inductively to conclude that $\flowin_l(\dens(t))$ is monotonically increasing in time and $\dot{\dens}_l\geq 0$ for all $l\in\Links$.%

Finally, observe that $\frac{\partial F_l}{\partial \dens_k}(\dens)\geq 0$ for all $k\neq l$ and $\dens\in\mathcal{A}$. Therefore, $\dot{\dens}=F(\dens)$ is \emph{cooperative} \cite{Hirsch:1985fk} within $\mathcal{A}$. For such systems, if $\tilde{\dens}(0)\leq \dens(0)\leq \bar{\dens}(0)$ then $\tilde{\dens}(t)\leq \dens(t)\leq \bar{\dens}(t)$ for all $t\geq 0$ where $\tilde{\dens}(t),\dens(t)$, and $\bar{\dens}(t)$ are trajectories of $\dot{\dens}=F(\dens)$. Taking $\tilde{\dens}(0)=0$ and $\bar{\dens}(t)\equiv \dense$, we have that $\tilde{\dens}(t)$ is monotonically increasing in time and bounded above by $\dense$ and therefore converges to an equilibrium. Corollary \ref{cor:freeflow_equil} implies $\dense$ is the unique equilibrium in $\mathcal{A}$, thus $\lim_{t\to\infty}\tilde{\dens}(t)=\dense$, concluding the proof.
\end{proof}

\subsection{Infeasible input flows}
We now wish to extend a notion of equilibrium to the case when the input flow is infeasible. We have already seen that the density of an ordinary link $\l\in\OLinks$ will not exceed the jam density $\densjam_l$ for any input flow. Thus any density accumulation due to the infeasible input flow must occur on the onramps $\Ramps$. It is therefore reasonable to consider an equilibrium condition in which the densities, input flows, and output flows on the ordinary links, and the output flows on onramp links, approach a steady state while onramp densities may grow without bound.
\begin{definition}
\label{def:equil2}
  For any input flow $\{\inflow_l\}_{l\in\Ramps}$, the collection $\{\flowe_l\}_{l\in\Links}$  is called an \emph{equilibrium flow} of the traffic network system if there exists a set $\{\dense_l\}_{l\in\Links}$ with 
  \begin{alignat}{4}
    \label{eq:23}
0&\leq \dense_l\leq \densjam_l &\ &\forall l\in\OLinks, \quad \text{and} \quad 0&\leq \dense_l\leq \infty &\ &\forall l\in\Ramps  
  \end{alignat}
such that
    \begin{alignat}{4}
      \label{eq:5}
      \flowe_l&=\flowout_l(\dense)=\flowin_l(\dense)&\ &\forall l\in\OLinks    \quad \text{and} \quad  \flowe_l&=\flowout_l(\dense)& \ &\forall l\in\Ramps
    \end{alignat}
and, for all $l\in\Ramps$, either $\flowe_l=\inflow_l$, or $\flowe_l< \inflow_l\text{ and }\dense_l=\infty$ where $\{\flowout_l(\dense)\}_{l\in\OLinks}$, $\{\flowin_l(\dense)\}_{l\in\OLinks}$, and $\{\flowout_l(\dense)\}_{l\in\Ramps}$ are determined by the PP/FIFO rule and we interpret $\fluxout_l(\infty)\triangleq \fluxoutmax_l$ for all $l\in\Ramps$. By a slight abuse of nomenclature, we call $\{\dense_l\}_{l\in\Links}$ an \emph{equilibrium density}.
\end{definition}

Definition \ref{def:equil2} naturally extends the definition for equilibrium flow given in Definition \ref{def:equil1} to the case when the input flow is infeasible. %

\begin{prop}
\label{prop:exist}
For constant input flows $\{d_l\}_{l\in\Ramps}$, an equilibrium flow exists.  
\end{prop}

The proof is provided in the appendix.

We now consider the uniqueness of equilibrium flows. We first consider the case when the traffic network graph is a polytree:

\begin{definition}
  A \emph{polytree} is a directed acyclic graph with exactly one undirected path between any two vertices. 
\end{definition}
Equivalently, a polytree is a weakly connected directed acyclic graph for which the underlying undirected graph contains no cycles. Figs. \ref{fig:nonmonotone1}(a), \ref{fig:nonmonotone}(a), and \ref{fig:nonmonotone}(b) depict polytrees.%

\begin{prop}
\label{prop:unique2}
  Given constant infeasible input flow $\{\inflow_l\}_{l\in\Ramps}$. If the traffic network graph $\network$ is a polytree, then the equilibrium flow $\{\flowe_l\}_{l\in\Links}$ is unique.
\end{prop}

\begin{proof}[Proof of Proposition \ref{prop:unique2}]

  Suppose there exists two equilibrium flows, $\{\flowe_l\}_{l\in\Links}$ and $\{\flowealt_l\}_{l\in\Links}$ with corresponding equilibrium densities $\{\dense_l\}_{l\in\Links}$ and $\{\densealt_l\}_{l\in\Links}$, and without loss of generality, assume $\flowealt_l<\flowe_l$ for a particular link $l$. If $l\in\OLinks$, by conservation of flow and fixed turn ratios \eqref{eq:9}, there must exist $k\in\Lin_{\tail(l)}$ such that $\flowealt_k<\flowe_k$. Continuing by induction, we conclude that there must exist $j\in\Ramps$ such that $\flowealt_j<\flowe_j$, thus we assume, without loss of generality, $l\in\Ramps$.

Observe that, since $l$ is an onramp and therefore $\flowe_l\leq \inflow_l$, we must have $\flowealt_l<\inflow_l$ and thus $\fluxout_l(\densealt_l)>\flowealt_l$. Furthermore, $\fluxout_l(\densealt_l)\geq \fluxout_l(\dense_l)$.

Let $l_1\triangleq l$. It is the case that for any link $l_i\in\Links$ for which $\fluxout_{l_i}(\densealt_{l_i})>\flowealt_{l_i}$, there must exist $l_{i+1}\in\Lout_{\head(l_i)}$ such that 
$\flowealt_{l_{i+1}}=\fluxin_{l_{i+1}}(\densealt_{l_{i+1}}),$ 
\emph{i.e.} there must exist an outgoing link with insufficient supply to meet the demand. 

Suppose, in this case, that $\flowealt_{l_{i+1}}<\flowe_{l_{i+1}}$. Since $\flowe_{l_{i+1}}\leq \fluxin_{l_{i+1}}(\dense_{l_{i+1}})$ by \eqref{eq:24}, we conclude $\fluxin_{l_{i+1}}(\densealt_{l_{i+1}})<\fluxin_{l_{i+1}}(\dense_{l_{i+1}})$ and since $\fluxin_{l_{i+1}}(\cdot)$ is decreasing, we must have $\densealt_{l_{i+1}}>\dense_{l_{i+1}}$. Therefore $\fluxout(\densealt_{l_{i+1}})>\fluxout(\dense_{l_{i+1}})$. Furthermore it must be $\fluxout_{l_{i+1}}(\densealt_{l_{i+1}})>\flowealt_{l_{i+1}}$, and we conclude there exists $l_{i+2}\in\Lout_{\head(l_{i+1})}$ such that $\flowealt_{l_{i+2}}=\fluxin_{l_{i+2}}(\densealt_{l_{i+2}})$. Continuing by induction, we create a sequence $l_1,l_2,\ldots,l_{i_0}$ until we reach link $l_{i_0}$ with $\flowealt_{l_{i_0}}<\flowe_{l_{i_0}}$ and $\fluxout_{l_{i_0}} (\densealt_{l_{i_0}})\geq \fluxout_{l_{i_0}} (\dense_{l_{i_0}})$ for which there exists $l_{i_0+1}\in\Lout_{\head(l_{i_0})}$ such that $\flowealt_{l_{i_0+1}}=\fluxin_{l_{i_0+1}}(\densealt_{l_{i_0+1}})$ but $\flowealt_{l_{i_0+1}}\geq\flowe_{l_{i_0+1}}$(Such a $i_0$ must exist since the traffic network contains no directed cycles). 

Now, due to fixed turn ratios, there must exist $l'\in\Lin_{\head(l_{i_0})}$ such that $\flowealt_{l'}>\flowe_{l'}$ and, by the PP/FIFO rule, 
\begin{align}
  \label{eq:6}
\fluxout_{l'}(\densealt_{l'})>\fluxout_{l'}(\dense_{l'}).  
\end{align}
Arguing as before, we thus conclude there exists $k_1\in\Ramps$ such that $\flowealt_{k_1}>\flowe_{k_1}$ for which $\fluxout_{k_1}(\dense_{k_1})\geq \fluxout_{k_1}(\densealt_{k_1})$. By a symmetric argument as above, we establish a sequence $k_1,\ldots,k_{j_0}$ until we reach link $k_{j_0}$ with $\flowealt_{k_{j_0}}>\flowe_{k_{j_0}}$ and $\fluxout_{k_{j_0}}(\dense_{k_{j_0}})\geq \fluxout_{k_{j_0}}(\densealt_{k_{j_0}})$ for which there exists $k_{j_0+1}\in\Lout_{\head(k_{j_0})}$ such that $\flowe_{k_{j_0+1}}=\fluxin_{k_{j_0+1}}(\dense_{k_{j_0+1}})$ but $\flowe_{k_{j_0+1}}\geq\flowealt_{k_{j_0+1}}$,  leading to the existence of $k'\in\Lin_{\head(k_{j_0})}$ such that $\flowe_{k'}>\flowealt_{k'}$ and $\fluxout_{k'}(\dense_{k'})>\flowout_{k'}(\densealt_{k'})$. Observe that we must have $k_{j_0}\neq l'$ by \eqref{eq:6}. 

 By another parallel argument, this then implies the existence of a link $m$ for which $\flowealt_m<\flowe_m$ and $\fluxout_m(\densealt_m)\geq \fluxout_m(\dense_m)$ for which there exists $m_1\in\Lout_{\head(m)}$ such that $\flowealt_{m_1}=\fluxin_{m_1}(\densealt_{m_1})$ and $\flowealt_{m_1}\geq \flowe_{m_1}$, and likewise we have $m\neq k'$. Furthermore, $m_1\neq l_{i_0}$, as this would imply a cycle in the underlying undirected graph. However, this process cannot continue indefinitely since the traffic network is finite.

\end{proof}
If the undirected traffic network does contain cycles, then equilibrium flows may not be unique when the input flow is infeasible. Such examples with nonunique equilibrium flows are not difficult to construct.%

In \cite{Coogan:2014ph}, we show that networks consisting of only merging junctions are cooperative. For example, a length of highway with onramps as in Fig. \ref{fig:nonmonotone}(a) is cooperative. Furthermore, we prove global convergence of flows under a mild restriction on $\beta_{lk}$ for each $l$, namely, that each incoming link $l$ at a particular junction routes the same fraction of flow off the network.  We include this result in the appendix for completeness.

\section{Ramp Metering}
\label{sec:ramp-metering}
Ramp metering is an active and rich area of research; see \cite{Papageorgiou:2000ys} for a review of approaches to ramp metering. In this section, we leverage the results on equilibria and convergence established above to design ramp metering strategies that achieve the maximum possible network throughput.
 \begin{definition}
   A \emph{ramp metering strategy} is a collection of functions $\{\meter_l(t)\}_{l\in\Ramps}$, $m_l(\cdot):\mathbb{R}_{\geq 0}\to\mathbb{R}_{\geq 0}$ that modifies the demand function of onramps. In particular, we introduce the \emph{metered demand function}
   \begin{equation}
     \label{eq:49}
     \fluxoutm_l(\dens_l(t))\triangleq \min\{\fluxout_l(\dens_l(t)),\meter_l(t)\} \quad \forall l\in\Ramps.
   \end{equation}
The traffic network dynamics are exactly as above with the metered demand function $\fluxoutm_l(\dens_l)$ replacing the demand function $\fluxout_l(\dens_l)$ for all $l\in\Ramps$.
 \end{definition}

In the following, we assume constant metering strategies, \emph{i.e.}, $\meter_l(t)\triangleq \meter_l$, and thus the network flow model remains autonomous for constant input flows. We then have $\fluxoutmmax_l\triangleq \min\{\fluxoutmax_l,m_l\}$ is the maximum outflow of link $l\in\Ramps$. The metering objective we consider is {network throughput} at equilibrium where \emph{network throughput} of an equilibrium flow $\{\flowe_l\}_{l\in\Links}$ is defined to be $\sum_{l\in\Ramps}\flowe_l$. %

The main result of this section is Theorem \ref{thm:ramp} which states that there exists a ramp metering strategy, obtained via a linear program, that induces an equilibrium with optimal network throughput. This result follows from the following proposition, which states that the resulting equilibrium flow from any ramp metering strategy is induced by a suitable choice of a (potentially different) ramp metering strategy and the assumption that each link is in freeflow. %

\begin{prop}
  \label{prop:ramp}
Consider a constant ramp metering strategy $\{\meter_l\}_{l\in\Ramps}$ that induces an equilibrium flow $\{\flowe_l\}_{l\in\Links}$ with a corresponding equilibrium density $\{\dense_l\}_{l\in\Links}$. Then there exists another constant ramp metering strategy $\{\meteralt_l\}_{l\in\Ramps}$ with the same equilibrium flow $\{\flowe_l\}_{l\in\Links}$ and new equilibrium density $\{\densealt_l\}_{l\in\Links}$ such that $\densealt_l\leq \densc_l$ for all  $l\in\OLinks$.

\end{prop}
\begin{proof}
  We construct such an alternative metering strategy explicitly.  For each $l\in\Ramps$, define $\meteralt_l\triangleq \flowe_l$ (if $\flowe_l=\inflow_l$, then we can in fact choose any $\meteralt_l\geq d_l$). If $d_l<\fluxoutmax_l$ and $\flowe_l=d_l$, let $\densealt_l$ be such that $\fluxout_l(\densealt_l)=\flowe_l$.  Otherwise, let $\densealt_l=\infty$.

For $l\in\OLinks$, let $\densealt_l\in[0,\densc_l]$ be such that $\fluxout_l(\densealt_l)=\flowe_l$ (such $\densealt_l$ always exists since $\flowe_l\leq \fluxc_l$ by Proposition \ref{prop:lessthancrit}). Observe that $\densealt_l\leq \dense_l$, and thus $\fluxin_l(\densealt_l)\geq \fluxin_l(\dense_l)$. It is easy to verify that $\{\flowe_l\}_{l\in\Links}$ satisfies the definition of an equilibrium flow for the metered networked flow system where $\alpha^v(\densealt)= 1$ for all $v\in\Verts$.
\end{proof}

Given infeasible demand $\{d_l\}_{l\in\Ramps}$, consider the following linear program:
\begin{alignat}{3}
  \label{eq:52}
  \max_{\{s_l\}_{l\in\Ramps},\{\flowe_l\}_{l\in\OLinks}}&\qquad&\sum_{l\in\Ramps}&s_l\\
\label{eq:52-2}\textnormal{subject to}&&\flowe_\OLinks&=\Routea\flowe_\OLinks+\Routeb s\\
&&0&\leq s_l\leq \min\{d_l,\fluxoutmax_l\}& \quad&\forall l\in\Ramps\\
\label{eq:52-mid}&&0&\leq \flowe_l\leq \fluxc_l&& \forall l\in\OLinks%
\end{alignat}
where $s=\begin{bmatrix}s_1&\ldots&s_{|\Ramps|}\end{bmatrix}'$. The feasible set \eqref{eq:52-2}--\eqref{eq:52-mid} is compact and thus the convex program attains its maximum. From a solution to \eqref{eq:52}--\eqref{eq:52-mid} we construct an optimal metering strategy:
      \begin{thm}
        \label{thm:ramp}
Let $\{s^\star_l\}_{l\in\Ramps}, \{{\flowe_l}^\star\}_{l\in\OLinks}$ be a maximizer of \eqref{eq:52}--\eqref{eq:52-mid}. Then any metering strategy $\{\meter_l\}_{l\in\Ramps}$ satisfying $m_l=s_l^\star$ if $s_l^\star<d_l$ and $m_l\geq s_l^\star$ if $s_l^*=d_l$
induces the equilibrium flow $\{\flowe_l\}_{l\in\Links}$ given by $\flowe_l={\flowe_l}^\star$ for all $l\in\OLinks$ and $\flowe_l=s_l^\star$ for all $l\in \Ramps$. Furthermore, $\{\flowe_l\}_{l\in\Links}$ achieves the maximum possible network throughput.
      \end{thm}
      \begin{proof}
Proposition \ref{prop:ramp} demonstrates the sufficiency of only considering ramp metering strategies that induce an equilibrium for which each ordinary link $l\in\OLinks$ is in freeflow. The program \eqref{eq:52}--\eqref{eq:52-mid} maximizes throughput subject to this free flow.
      \end{proof}

         We remark that, as in Proposition \ref{prop:roc}, $\lim_{t\to\infty}\dens_l(t)=\densealt_l$ for all trajectories such that $\dens_l(0)\leq \densealt_l$ for all $l\in\Links$ where $\{\densealt\}_{l\in\Links}$ is the equilibrium density construction in Proposition \ref{prop:ramp}.

    \begin{figure}
    \centering
\begin{tabular}{c c}
    \begin{tikzpicture}[yscale=.7]
\node[junc] (j1) at (0,0) {$v_1$};
\draw[onramp] ($(j1)+(180:.8in)$) -- node[above]{$1$} (j1);
\node[junc] (j2) at ($(j1)+(30:.8in)$) {$v_2$};
\node[junc] (j3) at ($(j1)+(-30:.8in)$) {$v_3$};
\node[junc, draw=gray] (j4) at ($(j2)+(0:.8in)$) {\color{gray}$v_4$};
\draw[onramp,gray] ($(j2)+(150:.8in)$) -- node[above]{\color{gray}$4$} (j2);
\draw[link,gray] (j2) -- node[above]{\color{gray}$5$} (j4);
\draw[link] (j1) -- node[above]{$2$}(j2);
\draw[link] (j1) -- node[above]{$3$}(j3);
    \end{tikzpicture} \\
(a)\\
\begin{tikzpicture}
  \begin{axis}[width=.4\textwidth,axis x line=center, axis y line=center, xmin=0, xmax=5, ymax=3.5,ymin=0,tick align=outside,     xtick={1,4}, xticklabel shift=-2pt, unit vector ratio*=1 .9 1, 
    xticklabels={$90$,$360$}, ytick={2,2.666},yticklabels={$3000$,$4000$},
    xlabel=$\dens_l$,
    ylabel=$\flux_l$,
legend style={at={(1.05,1)}, anchor=north west}
]
    \addplot+[mark=none, black] coordinates{
(0,0)
(1,2)
(5,2)
};
\addlegendentry{$\fluxout_l(\dens_l)$}
\addplot+[mark=none,black,dashed]coordinates{
(0,2.66)
(4,0)
};
\addlegendentry{$\fluxin_l(\dens_l)$}
  \end{axis}
\end{tikzpicture}\\
(b)
\end{tabular}
    \caption{A network with two onramps $\{1,4\}$ and three ordinary links $\{2,3,5\}$. The dark links and nodes are enough to demonstrate that the system is not cooperative due to the FIFO junction assumption: increased density on link 2 could restrict the outflow from link 1, and due to fixed split ratios, the inflow to link 3 also decreases. The additional shaded links and node illustrate how ramp metering increases network throughput: by restricting outflow from onramp 4, the density on link 2 may be reduced, leading ultimately to increased flow on link 3. (b) The supply and demand functions for links $\{2,3,5\}$. }
    \label{fig:nonmonotone1}
  \end{figure}
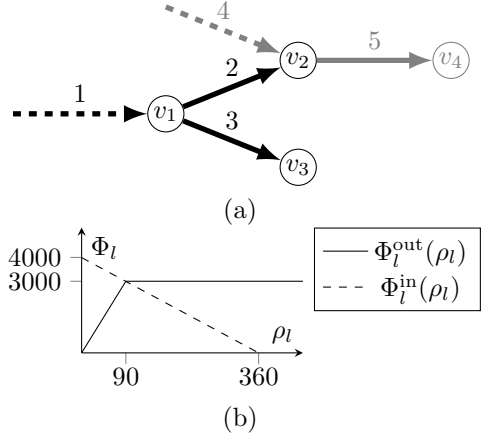
      \begin{example}
\label{ex:1revisit}
Consider a road network with two onramp links labeled $\{1,4\}$ and three ordinary links labeled $\{2,3,5\}$ as shown in Fig. \ref{fig:nonmonotone1}(a). This network is not cooperative; in particular, $\partial F_3(\dens)/\partial \dens_2 <0$ when link $2$ does not have adequate supply, see \cite[pp. 14--15]{Kurzhanskiy:2011fk} for a similar example. We assume links $\{2,3,5\}$ each have the supply and demand functions as shown in Fig. \ref{fig:nonmonotone1}(b) and that $\fluxoutmax_1=3000$ and $\fluxoutmax_4=6000$ (units are vehicles per hour), and we assume $\beta_{12}=\beta_{13}=\frac{1}{2}$ and $\beta_{25}=\beta_{45}=1$. With input flows $d_1=d_4=2500$, it can be verified that the equilibrium with no ramp metering is
\begin{alignat}{1}
  \label{eq:15}
\{\flowe_1,\flowe_2,\flowe_3,\flowe_4,\flowe_5\}&=\{2000,1000,1000,2000,3000\} \\
  \{\dense_1,\dense_2,\dense_3,\dense_4,\dense_5\}&=\{\infty,270,30,\infty,90\}
\end{alignat}
and therefore the total network throughput is $\flowe_1+\flowe_4=4000$. 

Solving \eqref{eq:52}--\eqref{eq:52-mid} and applying Theorem \ref{thm:ramp}, we choose $m_1\geq 2500, m_4=1750$, and then
\begin{alignat}{6}
  \label{eq:79}
\{\flowe_1,\flowe_2,\flowe_3,\flowe_4,\flowe_5\}&=\{2500,1250,1250,1750,3000\} \\
  \{\dense_1,\dense_2,\dense_3,\dense_4,\dense_5\}&=\{\infty,37.5,37.5,\infty,90\}  
\end{alignat}
with network throughput $\flowe_1+\flowe_4=4250$.%
      \end{example}

\section{Conclusions}
\label{sec:conc}
  We have proposed and analyzed a macroscopic traffic flow model that merges ideas from compartmental system theory and dynamical system theory with existing, validated traffic network models. We apply these results to develop a ramp metering strategy that optimizes throughput. Future work will investigate time-varying ramp metering strategies for time-varying onramp input demands. We will also consider alternative control objectives such as minimizing travel time or equalizing ramp queue length and alternative control methods such as route suggestion.

\appendix
\section{Proofs and Auxiliary Results}

\subsection{Proof of Proposition \ref{prop:exist}}
\label{sec:prop:exist}
\begin{proof}[Proof of Proposition \ref{prop:exist}]
We introduce a change of coordinates which allows us to capture ``equilibrium'' conditions in which onramps have infinite density.

Let 
\begin{equation}
  \label{eq:104}
t(x)\triangleq \frac{x}{1+x}, \quad t^{-1}(\hat{x})=\frac{\hat{x}}{1-\hat{x}}
\end{equation}
and define the change of coordinates
\begin{alignat}{2}
  \label{eq:97} %
T(\dens)&\triangleq\begin{bmatrix}t(\dens_1)&\ldots &t(\dens_{|\Ramps|})&\dens'_\OLinks\end{bmatrix}'\qquad& T^{-1}(\hdens)&=\begin{bmatrix} t^{-1}(\hdens_1)&\ldots & t^{-1}(\hdens_{|\Ramps|})&\hdens_\OLinks'\end{bmatrix}'
\end{alignat}
where we assume onramps $\Ramps$ are enumerated $1,\ldots,|\Ramps|$ and $\dens_\OLinks$ is the vector of densities for links $\OLinks$ for some enumeration of $\OLinks$.
Let $\hdens=\begin{bmatrix}\hdens_\Ramps'& \hdens_\OLinks'\end{bmatrix}'\triangleq T(\dens)$ and observe
\begin{alignat}{2}
\nonumber  \dot{\hdens}_l&=\left({1-\hdens_l}\right)^2F_l\left({\dens}\right)\\
&=\left({1-\hdens_l}\right)^2F_l\left(T^{-1}({\hdens})\right)\\
  \label{eq:97-2}&=:\hat{F}_l(\hdens)&\quad\forall l\in\Ramps.
\end{alignat}
Note that only onramps undergo a coordinate change, that is, $\dens_\OLinks=\hdens_{\OLinks}$.  

Similarly define 
\begin{equation}
  \label{eq:108}
\hat{F}_l(\hdens)\triangleq F_l(T^{-1}(\hdens))\quad \forall l\in\OLinks
\end{equation}
so that $\dot{\dens}_l=\dot{\hdens}_l=\hat{F}_l(\hdens)$ for all $l\in\OLinks$. Let $\hat{F}(\hdens)\triangleq \begin{bmatrix} \hat{F}_1(\hdens)&\ldots &\hat{F}_{|\Links|}(\hdens)\end{bmatrix}'$, then $\dot{\hdens}=\hat{F}(\hdens)$.

We introduce the change of coordinates so that $\hdens_l$ remains bounded even as $\dens_l\to\infty$ for $l\in\Ramps$. Furthermore, the definition of $\hat{F}(\hdens)$ can be suitably extended to the case where $\hdens_l=1$ for all $l\in\Ramps'$ for some subset $\Ramps'\subseteq \Ramps$, even though $T^{-1}(\hdens)$ is not defined for such $\hdens$. In particular, let
\begin{align}
  \label{eq:8}
\hDomain'&\triangleq \{\hdens: \hdens_l\in[0,1)\ \forall l\in\Ramps\text{ and }\hat{\dens}_l\in[0,\densjam_l]\ \forall l\in\OLinks\}
\end{align}
and
$\hDomain\triangleq \textbf{cl}(\hat{\Domain}')$, 
$S\triangleq \hDomain\backslash\hDomain'$
where $\textbf{cl}(\cdot)$ denotes closure. Observe that \eqref{eq:97-2} and \eqref{eq:108} define $\hat{F}(\cdot)$ on $\hat{\Domain}'$. Furthermore, Assumption 1 ensures that $\hat{F}(\cdot)$ is Lipschitz continuous on $\hDomain'$.
Lipschitz continuity implies uniform continuity, and thus there exists a unique extension of $\hat{F}(\cdot)$ to $\hDomain$ given by \cite[Chapter 4]{Rudin:1976qf}
\begin{align}
  \label{eq:119}
  \hat{F}(\hat{x})\triangleq\lim_{\substack{\hat{y}\to \hat{x}\\\hat{y}\in\hDomain'}}\hat{F}(\hat{x})\quad \forall \hat{x}\in S.
\end{align}
This definition is equivalent to interpreting $\fluxout_l(t^{-1}(1))\triangleq \fluxoutmax_l$ for $l\in\Ramps$ in the PP/FIFO rule. Existence and uniqueness of solutions to $\dot{\hdens}=\hat{F}(\hdens)$ initialized in $\hDomain$ follows readily.
We now define $\bF(\cdot):\hDomain\to\mathbb{R}^{|\Links|}$ as follows:
\begin{align}
  \label{eq:122}
  \bF(\hat{x})\triangleq
\left\{  \begin{aligned}
    &F(T^{-1}(\hat{x}))&&\text{if }\hat{x}\in\hDomain'\\
\lim_{\substack{\hat{y}\to\hat{x}, \hat{y}\in\hDomain'}}&\bF(\hat{y})&&\text{if }\hat{x}\in S
  \end{aligned}\right.
\end{align}
where again the limit is guaranteed to exist because $\bF(\cdot)$ is Lipschitz continuous on $\hDomain'$. 

We have $\hDomain$ is compact, convex, and positively invariant. It follows that there exists a stationary point $\hdens^\text{e}$ such that $\bar{F}(\hdens^\text{e})=0$ \cite[Theorem 4.20]{blanchini2008set}. This stationary point corresponds to an equilibrium density $\dense$ in the original coordinates via the map $T^{-1}(\cdot)$ in \eqref{eq:97} and thus gives an equilibrium flow as defined in Definition \ref{def:equil2}. In particular, $\hdens^\text{e}_l=1$ for $l\in \Ramps$ implies $\dense_l=\infty$.
\end{proof}

\subsection{Piecewise Differentiability}
\label{sec:pw_diff}
Assume $d_l(t)\equiv d_l$ for some constant $d_l$ for all $l\in\Ramps$ so that the dynamics are autonomous. Note that the solution of \eqref{eq:11-2}--\eqref{eq:11-3} is
\begin{align}
  \label{eq:26}
\textstyle &\alpha^v(\dens)=
\begin{cases}
\textstyle \min\left\{1,\min_{k\in\Lout_v}\left\{\left(\sum_{j\in\Lin_v}\turn^v_{jk} \fluxout_{j}(\dens_{j})\right)^{-1}\fluxin_k(\dens_k)\right\}\right\}\\
&\hspace*{-.67in}\text{if }\exists l\in\Lin_v\text{ s.t. }\dens_l>0\\
1&\hspace{0in}\text{otherwise,}
\end{cases}
\end{align}
thus $\{\flowout_{l}(\dens)\}_{l\in\Lin_v}$ is uniquely defined in \eqref{eq:11}. Furthermore, by considering the finite set of functions possible for $\alpha^v(\dens)$ determined by the minimizing $k\in\Lout_v$ in \eqref{eq:26}, we conclude that $\flowout_l(\dens)$ and thus $F_l(\dens)$ is a continuous selection of differentiable functions. Indeed, at each junction, either every outgoing link has adequate supply to accomodate the demand of incoming links, or there exists at least one link that does not have adequate supply. If more than one link does not have adequate supply, the most restrictive link determines the flow through the junction. Thus, for each $v\in\Verts$, there are $|\Lout_v|+1$ functions possible for $\alpha^v(\dens)$ in \eqref{eq:26}. We then consider $F(\dens)$ to be selected from $\prod_{v\in \Verts}(|\Lout_v|+1)$ \emph{modes} of the network.

Let $\mathcal{I}$ denote an index set of these possible modes, and let $F^{(i)}(\dens)$ for $i\in\mathcal{I}$ denote the particular mode defined implicitly by the corresponding minimizers of \eqref{eq:26} for each $v\in\Verts$. The function $F(\dens)$ is then {piecewise differentiable}. Let $J^{(i)}(\dens)$ denote the Jacobian of $F^{(i)}(\dens)$, which is well defined on $\{\dens\in\Domain^\circ\mid F(\dens)=F^{(i)}(\dens)\}$. Consider the directional derivative
\begin{align}
  \label{eq:141}
  F'(\dens_0;y)\triangleq \lim_{\substack{h\to 0\\h>0}}\frac{F(\dens_0+h y)-F(\dens_0)}{h}.
\end{align}
A key property of piecewise differentiable functions is that the derivative \eqref{eq:141} exists for all $\dens_0\in\Domain^\circ$ and $y\in\mathbb{R}^{|\Links|}$, and   $F'(\dens_0;y)\in\{J^{(i)}(\dens_0)y\mid i\in\mathcal{I}\}.$
It follows \cite{Scholtes:2012fk} that 
\begin{equation}
  \label{eq:116}
  \dot{F}(\dens)\triangleq \frac{d}{dt}F(\dens(t))\in\{J^{(i)}(\dens)F(\dens)\mid i\in\mathcal{I}\}.
\end{equation}
We use \eqref{eq:116} below when we consider stability of the traffic network.

\subsection{Cooperativity and Convergence in Networks with only Merging Junctions}

\begin{prop}[Proposition 9 of \cite{Coogan:2014ph}]
\label{prop:stab1}
    Given constant input flow $\{\inflow_l\}_{l\in\Ramps}$. 
\renewcommand{\labelenumi}{ (\arabic{enumi}) }
If 
    \begin{enumerate}
    \item $|\Lout_v|\leq 1$, for all $v\in\Verts$,
     \item For all $v\in\Verts$, there exists $\Gamma_v$ s.t. $\off_l=\Gamma_v$ $\forall l\in\Lin_v$
    \end{enumerate}
then there exists a unique equilibrium flow $\{\flowe_l\}_{l\in\Links}$ and
\begin{alignat}{2}
  \label{eq:56}
  \lim_{t\to\infty} \flowin_l(\dens(t))&=\flowe_l&&\forall l\in\OLinks\\
  \lim_{t\to\infty} \flowout_l(\dens(t))&=\flowe_l&\qquad&\forall l\in\Links
\end{alignat}
for any initial condition $\dens(0)\in\Domain$.

\end{prop}

The condition $|\Lout_v|\leq 1$, for all $v\in\Verts$ implies that each junction is a merging junction and consists of only one outgoing link or no outgoing links. The constraint $\off_l=\Gamma_v$ $\forall l\in\Lin_v$ for some $\Gamma_v$ implies that the fraction of flow exiting a link that is routed off the network is the same for each incoming link at a particular junction. To prove Proposition \ref{prop:stab1}, we introduce the following definition from \cite{Jacquez:1993uq}:%
\begin{definition}
\label{def:comp}
  A matrix $A\in\mathbb{R}^{n\times n}$ is a \emph{compartmental matrix} if $[A]_{ij}\geq 0$ for all  $i\neq j$ and $\sum_{i=1}^n[A]_{ij}\leq 0$ for all $j$
where $[A]_{ij}$ is the $ij$-th entry of $A$.
\end{definition}

Equivalently, $A$ is a compartmental matrix if and only if $A$ is \emph{Metzler} \cite{berman} and $\mu_1(A)\leq 0$ where $\mu_1(A) \triangleq \lim_{h\to 0^+}\frac{1}{h}({||I+hA||_1-1})$ is the \emph{logarithmic norm} of $A$ and $||A||_1$ is the matrix norm induced by the vector one-norm \cite{Desoer:1972uq}.  This observation provides a connection to \emph{contraction theory} for non-Euclidean norms \cite{Sontag:2010fk}. In particular, Lemma \ref{lem:main} below shows that $F(\dens)$ is nonexpansive in a region of the state-space relative to a weighted one-norm.

\begin{lemma}
\label{lem:main}
  Given ${\Omega}\subseteq \Domain$ and diagonal matrix $W$ with positive entries on the diagonal such that $WJ^{(i)}({\dens})$ is a compartmental matrix for all $i\in\mathcal{I}$ and all ${\dens}\in\Omega^\circ$ such that $F^{(i)}(\dens)=F(\dens)$ where $J^{(i)}(\dens)$ denotes the Jacobian of $F^{(i)}(\dens)$ and $\Omega^\circ$ denotes the interior of $\Omega$. Then $V(\dens)\triangleq ||WF(\dens)||_1$ is decreasing along trajectories $\dens(t)$ of the traffic network when $\dens(t)\in\Omega$. Moreover, if $\Omega$ is positively invariant, then the flows of the network converge to an equilibrium flow as defined in Definition \ref{def:equil2}.
\end{lemma}

\begin{proof}
The following proof is adapted from the proof of  \cite[Theorem 2]{Maeda:1978fk}. Consider the change of coordinates constructed in the proof of Proposition \ref{prop:exist} in Section \ref{sec:prop:exist}, and define $\bV(\cdot):\hDomain\to\mathbb{R}$ by
\begin{align}
  \label{eq:115}
\bV(\hdens)\triangleq ||W\bF(\hdens))||_1
\end{align}
where $\bF(\cdot)$ is given in \eqref{eq:122} so that $\bV(\hdens)=V(\dens)$ when $\hdens=T(\dens)$.  Let $\{F^{(i)}(\cdot)\mid i\in\mathcal{I}\}$ be the collection of modes as described in Section \ref{sec:pw_diff}.
It follows that
\begin{align}
  \label{eq:117}
  \dot{\bF}(\hat{x})\in\{\bJ^{(i)}(\hat{x})\bF(\hat{x})\mid i\in\mathcal{I}\}
\end{align}
where, defining $\hDomain^\circ$ to be the interior of $\hDomain$,
\begin{align}
  \label{eq:120}
  \bJ^{(i)}(\hat{x})\triangleq
  \begin{cases}
    J^{(i)}(T^{-1}(\hat{x}))&\text{if }\hat{x}\in\hDomain^\circ\\
\lim_{\hat{y}\to\hat{x},\hat{y}\in\hDomain^\circ}\bJ^{(i)}(\hat{y})&\text{if }\hat{x}\in \hDomain\backslash\hDomain^\circ.
  \end{cases}
\end{align}
By assumption and the above analysis, $W\bJ^{(i)}(\hat{x})$ is a compartmental matrix for all $\hat{x}\in\hOmega$ where $\hOmega\triangleq \{T(x):x\in\Omega\}$ and all $i$ such that $\bF(\hat{x})=\bJ^{(i)}(\hat{x})\bF(\hat{x})$, \emph{i.e.}, the selected index in \eqref{eq:117}.

There exists a vector $\nu(\hat{x})\in\{-1,0,1\}^{|L|}$ such that $V(\hat{x})=\nu(\hat{x})'W{\bF}(\hat{x})$ and $\dot{V}(\hat{x}(t))=\nu(\hat{x})'W\dot{\bF}(\hat{x})$ with the property
$\bF_l>0$ implies $\nu_l=1$ and $\bF_l<0$ implies $\nu_l=-1$.

We drop the supscript $(i)$ and time dependence $(t)$ notation for clarity. Let
\begin{align}
  \label{eq_new:1}
  I&=\{l\mid \bF_l >0, \text{ or }\bF_l=0, \dot{\bF}_l>0\}\\
  J&=\{l\mid \bF_l <0, \text{ or }\bF_l=0, \dot{\bF}_l<0\}\\
 K&=\{l\mid \bF_l =0 \text{ and }\dot{\bF}_l=0\}.
\end{align}
We partition $\bQ\triangleq W\bJ$ into blocks such that
\begin{align}
  \label{eq_new:2}
  \begin{bmatrix}
   W_I\dot{\bF}_I\\
   W_J\dot{\bF}_J\\
    W_K\dot{\bF}_K
  \end{bmatrix}
=
\begin{bmatrix}
  \bQ_{II}&  \bQ_{IJ} &  \bQ_{IK}\\
  \bQ_{JI}&  \bQ_{JJ} &  \bQ_{JK}\\
  \bQ_{KI}&  \bQ_{KJ} &  \bQ_{KK}\\
\end{bmatrix}
\begin{bmatrix}
  \bF_I\\
  \bF_J\\
  \bF_K
\end{bmatrix}
\end{align}
where $\bQ_{IJ}=[q_{ij}]_{i\in I,j\in J}$, $\bF_I=\{\bF_l\}_{l\in I}$, \emph{etc.} and $W_I, W_J, W_K$ are the diagonal blocks of $W$. Then
\begin{align}
  \label{eq_new:3}
  \dot{\bV}=&\mathbf{1}^TW_I\dot{\bF}_I-\mathbf{1}^TW_J\dot{\bF}_J\\
  \label{eq_new:3-2}=&\mathbf{1}^T\bQ_{II}\bF_I+\mathbf{1}^T\bQ_{IJ}\bF_J-\mathbf{1}^T\bQ_{JI}\bF_I-\mathbf{1}^T\bQ_{JJ}\bF_J\\
\nonumber =&-(2\mathbf{1}^T\bQ_{JI}+\mathbf{1}^T\bQ_{KI}+\alpha_I)\bF_I\\
\label{eq_new:4-2}&+(2\mathbf{1}^T\bQ_{IJ}+\mathbf{1}^T\bQ_{KJ}+\alpha_J)\bF_J
\end{align}
where
\begin{align}
  \label{eq_new:4}
  \alpha_I&=-(\mathbf{1}^T\bQ_{II}+\mathbf{1}^T\bQ_{JI}+\mathbf{1}^T\bQ_{KI})\\
  \alpha_J&=-(\mathbf{1}^T\bQ_{IJ}+\mathbf{1}^T\bQ_{JJ}+\mathbf{1}^T\bQ_{KJ})
\end{align}
and $-\alpha_I\leq 0$ and $-\alpha_J\leq 0$ (where $\leq$ is interpreted elementwise) because the column sums of $\bQ$ are less than or equal to zero since $\bQ$ is assumed to be a compartmental matrix. Note that, additionally, the entries of $\bQ_{IJ}$, $\bQ_{JI}$, $\bQ_{IK}$, $\bQ_{KI}$, $\bQ_{JK}$, and $\bQ_{KJ}$ are nonnegative since $\bQ$ is nonnegative for entries not on the diagonal. It then follows that \eqref{eq_new:4-2} is nonpositive, \emph{i.e.},   $\dot{\bV}\leq 0$.

Supposing $\Omega$ is positively invariant, we show convergence to an equilibrium flow via LaSalle's invariance principle. To that end, define $E\triangleq\{\hat{x}:\dot{\bV}(\hat{x})=0\}$.  Let $\hat{z}(t)$ be a trajectory completely contained in $E$ and consider $L_+(t)=\{l\mid \bF_l(\hat{z}(t))>0\}$ and $L_-(t)=\{l\mid \bF_l(\hat{z}(t))<0\}$. We have $\bar{V}(\hat{z})\equiv V_\infty$ for some $V_\infty$, and thus 
\begin{equation}
  \label{eq:93}
\mathbf{1}'W_{L_+}\bF_{L_+}-\mathbf{1}'W_{L_-}\bF_{L_-}\equiv V_\infty.  
\end{equation}
We now claim the sets $L_+(t)$ and $L_-(t)$ are monotonically increasing with respect to set inclusion. %
To prove the claim, observe that since $\dot{\bV}\equiv 0$ and considering \eqref{eq_new:3-2}, we have $\mathbf{1}^T\bQ_{JI}\bF_I\equiv 0$ and $\mathbf{1}^T\bQ_{IJ}\bF_J\equiv 0$. This implies
\begin{align}
  \label{eq_new:5}
  \sum_{k\in I}q_{lk}\bF_k\equiv 0 \quad \text{for all }l\in J\\
  \sum_{k\in J}q_{lk}\bF_k\equiv 0 \quad \text{for all }l\in I
\end{align}
since the entries of $\bQ_{JI}$, $\bQ_{IJ}$, $\bF_I$ are nonnegative and the entries of $\bF_J$ are nonpositive and thus each entry of $\bQ_{JI}\bF_I$ and $\bQ_{IJ}\bF_J$ must be zero.
 Since $L_+\subset I$ and $L_-\subset J$, and $\bF_l=0$ for all $l\in I\backslash L_+$ and for all $l\in J\backslash L_-$, we have
 \begin{align}
   \label{eq_new:6}
   \sum_{k\in L_+}q_{lk}\bar{F}_k\equiv  0\text{ for all }l\in L_-\\
\sum_{k\in L_-}q_{lk}\bar{F}_k\equiv  0\text{ for all }l\in L_+.
 \end{align}
Consider $l\in L_+(\tau)$ at some time $\tau$, and from \eqref{eq_new:2}, we have
\begin{align}
  \label{eq:121}
  W_l\dot{\bF}_l&=q_{ll}\bF_l+\sum_{k\in L_+\backslash \{l\}}q_{lk}\bF_k+\sum_{k\in L_-}q_{lk}\bF_k\\
&=q_{ll}\bF_l+\sum_{k\in L_+\backslash \{l\}}q_{lk}\bF_k\\
&\geq q_{ll}\bF_l 
\end{align}
where $q_{ll}\leq 0$. It follows that since $\bF_l(\hat{z}(\tau))>0$, then
\begin{align}
  \label{eq:1}
 \bar{F}_l(\hat{z}(t))\geq \bar{F}_l(\hat{z}(\tau))e^{(t-\tau)q_{ll}/W_l}>0
\end{align}
for all $t\geq \tau$, and thus $\bF_l(\hat{z}(t))>0$ for all $t\geq \tau$, implying that $L_+$ is monotonically increasing. A similar analysis holds for $L_-(t)$, proving the claim.

Furthermore, because $\dot{\bV}(\hat{z})\equiv 0$, we must have $\alpha_I\bar{F}_I\equiv 0$ and $\alpha_J\bar{F}_J\equiv 0$ since, in \eqref{eq_new:4-2}, each term $\bar{Q}_{JI}$, $\bar{Q}_{KI}$, $\alpha_I$, $\bar{Q}_{IJ}$, $\bar{Q}_{KJ}$, and $\alpha_{J}$ are nonnegative. Since $\mathbf{1}^T\bar{Q}\bar{F}=-(\alpha_I\bar{F}_I+\alpha_J\bar{F}_J)$, we have that $\mathbf{1}^T\bar{Q}\bar{F}=\mathbf{1}'W\bJ\bF_l\equiv 0$ and thus $\mathbf{1}'\dot{\bF}\equiv 0$. Therefore
\begin{equation}
  \label{eq:95}
\mathbf{1}'W_{L_+}\bF_{L_+}+\mathbf{1}'W_{L_-}\bF_{L_-}\equiv C
\end{equation}
for $C$ some constant.

Combining \eqref{eq:93} and \eqref{eq:95}, we have that $2\mathbf{1}'W_{L_-}\bF_{L_-}\equiv C-V_\infty$. If $C-V_\infty\neq 0$, then $\int_0^\infty \mathbf{1}'W_{L_-}\bar{F}_{L_-} \ dt=-\infty$, but this is  a contradiction since $\dot{z}=\bF(\hat{z})$ where $z(t)\triangleq T^{-1}(\hat{z})$ and $z$ is bounded below. Thus $L_-(t)\equiv \emptyset$. 

Let $\Ramps^\infty\triangleq \{l\in\Ramps\mid \exists c>0 \text{ s.t. }\bF_l(\hat{z})\equiv c\}$, and let $M_+\triangleq L_+\backslash \Ramps^\infty$. Since $2\mathbf{1}'W_{L_+}\bF_{L_+}=V_\infty+C$, it follows that $\mathbf{1}'W_{M_+}\bF_{M_+}(\hdens)\equiv C_2$ for some constant $C_2\geq 0$.
If $C_2> 0$, then $\int_0^{-\infty}\mathbf{1}'W_{M_+}\bar{F}_{M_+} \ dt=\infty$, which is also a contradiction since with  $y(t)\triangleq z(-t)$ and $\hat{y}(t)\triangleq \hat{z}(-t)$, we have $\dot{y}_l=-\bF_l(\hat{y})$ for all $l\in M_+$ and $y_l(t)$ is bounded below. Therefore $C_2=0$, and we have shown $\bar{F}_l(\hat{z})\equiv 0$ for all $l\in \Links\backslash \Ramps^\infty$. Combined with the definition of $\Ramps^\infty$, this implies that $\hat{z}(t)\equiv \hat{z}^\text{e}\in \hat{\Domain}$. Furthermore, these are exactly the conditions required such that $z^\text{e}\triangleq T^{-1}(\hat{z}^\text{e})$ is an equilibrium density as defined in Definition \ref{def:equil2}.

\end{proof}

We now turn our attention to the class of networks considered in Proposition \ref{prop:stab1}. For networks satisfying condition 1 of Proposition \ref{prop:stab1}, $\Sink$ is a singleton, suppose $\Sink=\{v_\text{sink}\}$. Furthermore, for each $l\in\Links$ there exists a unique path $\{l_1,\ldots,l_{n_l}\}\subset\Links$ with $l_1=l$ such that $\head(l_{n_l})=v_\text{sink}$. Supposing 1) and 2) of Proposition \ref{prop:stab1}, let 
  \begin{align}
    \label{eq:74}
    w_l&\triangleq
    \begin{cases}
      1-\Gamma_{\head(l)}&\text{if }\Gamma_{\head(l)}<1\\
      1&\text{if }\Gamma_{\head(l)}=1
    \end{cases}\qquad \forall l\in\Links\\
  \label{eq:113}
 W_l&\triangleq w_{l_1}\cdot\ldots\cdot w_{l_{n_l}}. 
\end{align}

\begin{lemma}
\label{lem:2}
  Given a traffic network with constant input flows $\{\inflow_l\}_{l\in\Ramps}$ satisfying the 1) and 2) of Proposition \ref{prop:stab1}. Define $w_l$ and $W_l$ as in \eqref{eq:74}--\eqref{eq:113}, and let $W\triangleq \text{diag}(W_1,\ldots,W_{|\Links|})$. Then $  W\left(\frac{\partial F^{(i)}}{\partial \dens}(\dens)\right)$
is a compartmental matrix for all $i\in\mathcal{I}$ such that $F(\dens)=F^{(i)}(\dens)$ and $\dens\in\Domain^\circ$.

\end{lemma}

\begin{proof}

Consider a particular link $l$ and the corresponding $l$th column of $\partial F^{(i)}/\partial \dens$ for some $i\in\mathcal{I}$. In the following, we omit the superscript $(i)$ and all partial derivatives are assumed to correspond to the mode $i$. We have
\begin{align}
\nonumber\textstyle\sum_{k\in\Links}W_k \frac{\partial F_k}{\partial \dens_l}=\frac{\partial}{\partial \dens_l}\Big({\textstyle-\sum_{k\in\Lin_{\tail(l)}}W_k\flowout_k+W_l\flowin_l} \\
  \label{eq:94}\textstyle{\textstyle-\sum_{k\in\Lin_{\head(l)}}W_k\flowout_k+\sum_{k\in\Lout_{\head(l)}}W_k\flowin_k}\Big).
\end{align}

It can be shown that, for networks such that $|\Lout_v|\leq 1$ for all $v\in\Verts$, we have $\frac{\partial F_k}{\partial \dens_l}\geq 0$ for all $l\neq k$, \emph{i.e.}, the system is \emph{cooperative}  \cite{Hirsch:1985fk}. Observe that  $\turn^{\tail(l)}_{kl}=(1-\Gamma_{\tail(l)})$ and $W_k=(1-\Gamma_{\tail(l)})W_l$ for all $k\in\Lin_{\tail(l)}$ for all $l\in\OLinks$. We subsequently show
\begin{align}
  \label{eq:111}
\textstyle \frac{\partial}{\partial \dens_l}\left(-\sum_{k\in\Lin_{\tail(l)}}W_k\flowout_k(\dens) +W_l\flowin_l(\dens)\right)(x)&=0\\
  \label{eq:111-2}\textstyle \frac{\partial}{\partial \dens_l}\left(-\sum_{k\in\Lin_{\head(l)}}W_k\flowout_k(\dens)+\sum_{k\in\Lout_{\head(l)}}W_k\flowin_k(\dens)\right)(x)&\leq 0
\end{align}
for all $x\in\Domain^\circ$. Combining \eqref{eq:111}--\eqref{eq:111-2} with \eqref{eq:94} gives $\sum_{k\in\Links}W_k \frac{\partial F_k}{\partial \dens_l}\leq 0$ for all $l$,  thus proving the claim. To prove \eqref{eq:111}--\eqref{eq:111-2}, consider a particular $x\in\Domain^\circ$:

\noindent\textbf{(Flows at $\tail(l)$)}
If upstream demand exceeds the supply of link $l$, that is, $l\in \OLinks$ and $\fluxin_l(x_l)<\sum_{k\in\Lin_{\tail(l)}}\beta^{\tail(l)}_{kl}\fluxout_k(x_k)$, then the PP/FIFO rule stipulates
\begin{align}
  \label{eq:96}
  \flowout_k(y)&=\fluxin_l(y_l)\frac{\fluxout_k(y_k)}{\sum_{j\in\Lin_{\tail(l)}}\beta^{\tail(l)}_{jl}\fluxout_j(y_j)}\ \forall k\in\Lin_{\tail(l)}\\
\flowin_l(y)&=\fluxin_l(y_l)
\end{align}
for all $y\in \Ball_\epsilon(x)$ for some $\epsilon>0$ where $\Ball_\epsilon(x)$ is the ball of radius $\epsilon$ centered at $x$. Then $\sum_{k\in\Lin_{\tail(l)}}\flowout_k(y)=\left(1-\Gamma_{\tail(l)}\right)^{-1}\fluxin_l(y_l)$ and %
\begin{align}
  \label{eq:99}
\textstyle       -\sum_{k\in\Lin_{\tail(l)}}W_k \flowout_k(y)+W_l\flowin_l(y)%
&=0\ \forall y\in\Ball_\epsilon(x)
\end{align}
which implies \eqref{eq:111}. 

If link $l$ has adequate supply, we have $\flowout_k(x)=\fluxout_k(x_k)$ for $k\in\Lin_{\tail(l)}$ and $\flowin_l(x)=\sum_{k\in\Lin_{\tail(l)}}\beta_{kl}^{\tail(l)}\fluxout(x_k)$, neither of which is a function of $x_l$, and thus also \eqref{eq:111} holds.

\noindent\textbf{(Flows at $\head(l)$)} By hypothesis, $\Lout_{\head(l)}$ is either empty or a singleton. If it is empty, then $\flowout_k(x)=\fluxout_k(x_k)$ for all $k\in\Lin_{\head(l)}$ and the lefthand side of \eqref{eq:111-2} is
\begin{align}
  \label{eq:110}
\textstyle-\sum_{k\in\Lin_{\head(l)}}W_k\frac{\partial}{\partial \dens_l}\flowout_k(x)&=\textstyle-W_l\fluxoutp_l(x_l)<0
\end{align}
and \eqref{eq:111-2} holds. If $\Lout_{\head(l)}$ is nonempty, let $\Lout_{\head(l)}=\{m\}$ and observe that $W_k=(1-\Gamma_{\head(l)})W_m$ and $\beta_{km}^{\head(l)}=1-\Gamma_{\head(l)}$ for all $k\in\Lin_{\head(l)}$. Suppose link $m$ has adequate supply for upstream demand so that  $\flowout_k(x)=\fluxout_k(x_k)$ for all $k\in\Lin_{\head(l)}$ and $\flowin_m(x)=\sum_{k\in\Lin_{\head(l)}}\turn^{\head(l)}_{km}\fluxout_k(x_k)$.
Then the lefthand side of \eqref{eq:111-2} is
$-W_l\fluxoutp_l(x_l) +W_m\beta_{lm}^{\head(l)}\fluxoutp_l(x_l)=0$
and therefore \eqref{eq:111-2} holds. If link $m$ has inadequate supply, there exists $\epsilon>0$ such that
\begin{alignat}{2}
  \label{eq:102}
\textstyle    \sum_{k\in\Lin_{\head(l)}}\flowout_k(y)&=\left(1-\Gamma_{\head(l)}\right)^{-1}&&\fluxin_m(y_m)\quad\\
&&&\qquad \forall y\in\Ball_\epsilon(x)\\
\flowin_m(y)&=\fluxin_m(y_m) \hspace{.3in} &&\qquad \forall y\in\Ball_\epsilon(x).
\end{alignat}
Then
$-\sum_{k\in\Lin_{\head(l)}}W_k\flowout_k(y)+W_m \flowin_m(y)=0$ for all $ y\in\Ball_\epsilon(x)$
and \eqref{eq:111-2} follows.
\end{proof}

\label{appen:proofprop9}
\begin{proof}[Proof of Proposition \ref{prop:stab1}]
  A network satisfying the 1) and 2) of the proposition consists of only merging junctions and is necessarily a polytree, thus Proposition \ref{prop:unique2} ensures uniqueness of the equilibrium flow.

By Lemma \ref{lem:2} above, $WJ^{(i)}(\dens)$ is a compartmental matrix for all $i\in\mathcal{I}$ such that $F(\dens)=F^{(i)}(\dens)$ and all $\dens\in\Domain^\circ$. Applying Lemma \ref{lem:main} with $\Omega \triangleq \Domain$ completes the proof.

\end{proof}

\bibliographystyle{ieeetr}
\bibliography{$HOME/Documents/Books/books}

\end{document}